\documentclass[letterpaper, 9pt, conference]{ieeeconf}

\IEEEoverridecommandlockouts

\usepackage{graphics}
\usepackage{epsfig}
\usepackage{times}
\usepackage{amsmath}
\usepackage{amssymb}
\usepackage{subcaption}
\usepackage{url}
\usepackage{bbm, dsfont}
\usepackage{mathtools}
\usepackage{multirow}
\usepackage{cite}
\usepackage{color}
\usepackage[normalem]{ulem}
\usepackage{algorithm}
\usepackage{algpseudocode}
\usepackage{tikz-cd}

\newtheorem{theorem}{Theorem}
\newtheorem{lemma}[theorem]{Lemma}
\newtheorem{proposition}[theorem]{Proposition}
\newtheorem{definition}[theorem]{Definition}
\newtheorem{remark}[theorem]{Remark}

\title{\LARGE \bf
Linear Systems as Representations of Time Groups}

\author{Subhrajit Sinha\\
\thanks{S. Sinha is with Pacific Northwest National Laboratory, Richland, WA 99354 USA. email: \texttt{subhrajit.sinha@pnnl.gov}.}
}

\begin{document}

\maketitle
\thispagestyle{empty}
\pagestyle{empty}

\begin{abstract}
In this paper, we develop a representation-theoretic formulation of
discrete-time linear systems. We show that such systems are naturally viewed as
representations of time groups acting on vector spaces, thereby endowing the
state space with a canonical algebraic structure. This perspective provides a
unified framework for linear systems over different fields, in which familiar
structural properties arise from the underlying representation. In particular,
invariant decompositions of the state space correspond to invariant
subrepresentations, while the distinctions between real, complex, and
finite-field systems emerge from the algebraic properties of the base field and
the time group. We further show that linear systems over finite fields
naturally correspond to representations of finite cyclic time groups, leading to
module structures over polynomial quotient rings. This provides a systematic
alternative to spectral analysis in settings where eigenvalue-based methods are
not the most natural organizing language.
\end{abstract}

\section{Introduction}\label{sec_intro}

A discrete-time linear system is classically described by the equation
\begin{eqnarray*}
x_{t+1} = A x_t,
\end{eqnarray*}
where $x_t$ is the state of the system and $A$ is the state transition matrix
\cite{kailath_linear_book}. This formulation is standard and extremely useful and it
captures the evolution of the system through repeated application of a linear
map, so that for an initial condition $x_0$, the state at time $t$ is given by
$x_t = A^t x_0$.

At the same time, this classical description conceals an important structural
point. The matrix $A$ is not itself an intrinsic object: it depends on a choice
of basis for the state space, and different choices of coordinates lead to
different matrices related by similarity transformations while describing the
same underlying dynamics \cite{kailath_linear_book}. Thus, although linear systems
are typically introduced in terms of matrices, the system itself should not be
identified with any particular matrix realization.

What is intrinsic is the action of time on the state space. Indeed, the dynamics
are determined not merely by the single matrix $A$, but by the family
$\{A^t\}_{t \in \mathbb{Z}}$ (or $\mathbb{N}$ in the noninvertible case), which
encodes how the system evolves over time. This naturally leads to the observation
that a discrete-time linear system defines a homomorphism
\begin{eqnarray}
\rho : G \to GL_k(V),
\end{eqnarray}
where $G$ represents the algebraic structure of time and $V$ is the state space
over a field $k$. In the classical invertible case, $G = \mathbb{Z}$; in periodic
settings, one is naturally led to finite cyclic groups. From this perspective, a
linear system is best understood not as a matrix, but as a linear representation
of time.

This shift in viewpoint has several consequences. First, it yields an intrinsic,
coordinate-free formulation of linear state evolution. Second, it shows that the
algebraic structure associated with linear systems is not auxiliary, but natural:
by a standard result from representation theory, such systems are equivalently
modules over the group algebra $k[G]$
\cite{serre_linear_representation_book,dummit_foote,lang_algebra_book}. Thus,
familiar constructions from linear systems theory may be reinterpreted as
structural features of a representation of time. Third, it makes clear that the
behavior of a linear system depends not only on the evolution law, but also on
the underlying field. Over $\mathbb{R}$ and $\mathbb{C}$, this leads to the
familiar spectral picture; over finite fields, the same representation-theoretic
structure remains intact even when spectral methods are no longer the most natural
organizing language \cite{lidl_niederreiter_finite_fields}.

This perspective should also be contrasted with the behavioral approach to
systems theory \cite{willems_behavioral_1986,polderman_willems,willems_polderman_book,oberst_multidimensional}. In the behavioral framework, a linear dynamical system is defined intrinsically as a shift-invariant linear subspace of trajectory space, and the emphasis is placed on the set of admissible trajectories rather than on a state evolution law. The present work is closer in spirit to the state-space viewpoint, but seeks to identify the intrinsic object underlying that formulation. In this sense, the contribution here is not a behavioral definition of a system, but a representation-theoretic formulation of linear state evolution.

The purpose of this paper is to make this representation-theoretic structure
explicit and to develop some of its consequences for linear systems theory. We
show that classical structural results arise naturally within this framework, with
invariant decompositions corresponding to invariant subrepresentations and with
the distinction between real, complex, and finite-field systems emerging from the
algebraic properties of the field and the time group. In particular, we show that
linear systems over finite fields fit naturally into the same framework, leading
to module structures over polynomial quotient rings and providing a systematic
alternative to purely spectral descriptions.

Taken together, these observations suggest that, at a more fundamental level, the
study of linear systems is the study of linear representations of time groups.

\section{Revisiting Linear Systems}\label{sec_linear_system}

Consider a deterministic discrete-time linear system \cite{kailath_linear_book}
\begin{eqnarray}\label{linear_sys}
x_{t+1} = Ax_t,
\end{eqnarray}
where $x_t \in \mathbb{R}^n$ is the state at time $t \in \mathbb{Z}$, and
$A \in M_n(\mathbb{R})$ is the state transition matrix. For the moment, assume
$A \in GL(\mathbb{R}^n)$ so that the dynamics are invertible. This is the
standard state-space description of a linear system. Since it is expressed in
terms of a matrix, it depends on a choice of basis for the state space and is
therefore not intrinsic.

Given an initial condition $x_0$, the state at time $t$ is $x_t = A^t x_0$. Thus, the evolution of the system is completely described by the family
$\{A^t\}_{t \in \mathbb{Z}}$. This defines a map
\begin{eqnarray}\label{rho_first_pass}
\rho : \mathbb{Z} \to GL(\mathbb{R}^n), \qquad t \mapsto A^t,
\end{eqnarray}
and the developments in this paper begin with the observation that a linear
system naturally gives rise to the map $\rho$ in \eqref{rho_first_pass}.

\begin{lemma}\label{lemma_rep_def}
The map $\rho$ in \eqref{rho_first_pass} satisfies
\begin{eqnarray*}
\rho(t_1 + t_2) = \rho(t_1)\rho(t_2)
\end{eqnarray*}
for all $t_1,t_2 \in \mathbb{Z}$, and $\rho(0) = I_n$.
\end{lemma}

\begin{proof}
For $t_1,t_2 \in \mathbb{Z}$,
\begin{eqnarray*}
\rho(t_1 + t_2) = A^{t_1 + t_2} = A^{t_1} A^{t_2} = \rho(t_1)\rho(t_2),
\end{eqnarray*}
and $\rho(0) = A^0 = I_n$.
\end{proof}

Lemma \ref{lemma_rep_def} shows that the family $\{A^t\}_{t\in\mathbb{Z}}$ is not
an arbitrary collection of operators, but is organized by the additive structure
of time. In particular, the evolution over a combined time interval is obtained
by composing the evolutions over the constituent intervals. Conversely, any map
$\rho : \mathbb{Z} \to GL(\mathbb{R}^n)$ satisfying
$\rho(t_1+t_2)=\rho(t_1)\rho(t_2)$ is completely determined by the single
operator $\rho(1)$, since $\rho(t)=\rho(1)^t$ for all $t \in \mathbb{Z}$.
Thus, the dynamics of an invertible linear system are fully encoded in how the
unit time step acts on the state space.

Thus, the matrix $A$ may be viewed not merely as a linear operator, but as the
generator of a time-indexed family of operators satisfying a compatibility
condition with the additive structure of time. At the same time, this point of
view shows that the matrix itself is not the intrinsic object of the theory.
Different choices of basis lead to different matrix realizations of the same
underlying dynamics. In this sense, the object of primary interest is not $A$ in
isolation, but the time action encoded by $\rho$. This shift in perspective
reveals that the dynamics already carry an underlying algebraic structure, and
it is this structure that will be developed in the remainder of the paper.

\begin{remark}
If $A$ is not invertible, the same construction yields a map
$t \mapsto A^t$ from the semigroup $(\mathbb{N},+)$ into
$\mathrm{End}(\mathbb{R}^n)$ satisfying the same composition property.
In this case, the evolution is still governed by the additive structure of time,
but without the requirement of invertibility.
\end{remark}

\section{From Linear Systems to Group Actions}\label{sec_group_action}

The map $\rho$ in Lemma \ref{lemma_rep_def} is elementary, but in this work we
take it as the structural starting point and once the evolution of a linear system
is viewed through the map
\begin{eqnarray*}
\rho : \mathbb{Z} \to GL(\mathbb{R}^n), \qquad t \mapsto A^t,
\end{eqnarray*}
the natural mathematical language is that of group actions and representation
theory.

In general, representation theory studies how groups act on structured spaces
\cite{serre_linear_representation_book,lang_algebra_book,rotman_group_theory_book}. Let $G$ be a group and $X$
a set. A left action of $G$ on $X$ is a map $\phi:G\times X\to X$ such that
\begin{eqnarray*}
g\cdot(h\cdot x)=(gh)\cdot x,\qquad e\cdot x=x.
\end{eqnarray*}
Equivalently, such an action is described by a homomorphism
\begin{eqnarray}
\rho:G\to \mathrm{Aut}(X),
\end{eqnarray}
where $\mathrm{Aut}(X)$ denotes the structure-preserving bijections of $X$.

\begin{definition}[Representation of a Group \cite{lang_algebra_book,rotman_group_theory_book}]
A representation of a group $G$ on a space $X$ is a homomorphism
\begin{eqnarray*}
\rho:G\to \mathrm{Aut}(X).
\end{eqnarray*}
\end{definition}

If $X=V$ is a vector space, then $\rho:G\to GL(V)$ is called a linear
representation. Thus, linear representation theory studies group actions on
vector spaces via invertible linear maps
\cite{serre_linear_representation_book,dummit_foote,lang_algebra_book}.

From this point of view, \emph{the discrete-time system \eqref{linear_sys} is
precisely an action of the time group $\mathbb{Z}$ on $\mathbb{R}^n$ by linear
transformations.}

\begin{proposition}\label{prop_group_algebra_module}
Let $G$ be a group, let $k$ be a field, and let $V$ be a vector space over $k$.
Then a linear representation
\begin{eqnarray}
\rho:G \to GL(V)
\end{eqnarray}
is equivalent to giving $V$ the structure of a left $k[G]$-module, where
$k[G]$ denotes the group algebra of $G$ over $k$.
\end{proposition}

\begin{proof}
Suppose first that
\begin{eqnarray}
\rho:G \to GL(V)
\end{eqnarray}
is a linear representation of $G$ on $V$. Recall that the group algebra
$k[G]$ consists of all finite formal sums
\begin{eqnarray}
\sum_{g \in G} a_g g, \qquad a_g \in k,
\end{eqnarray}
with multiplication induced by the group multiplication in $G$ and extended
linearly.

Define an action of $k[G]$ on $V$ by
\begin{eqnarray}
\left(\sum_{g \in G} a_g g\right)\cdot v
=
\sum_{g \in G} a_g \rho(g)v,
\qquad v \in V.
\end{eqnarray}
This is well-defined and linear in both the coefficients and the vector.
Moreover, if
\begin{eqnarray}
\alpha = \sum_{g \in G} a_g g, \qquad
\beta = \sum_{h \in G} b_h h,
\end{eqnarray}
then for any $v\in V$,
\begin{eqnarray*}
(\alpha\beta)\cdot v
&=&
\left(\sum_{g,h \in G} a_g b_h (gh)\right)\cdot v \\
&=&
\sum_{g,h \in G} a_g b_h \rho(gh)v \\
&=&
\sum_{g,h \in G} a_g b_h \rho(g)\rho(h)v \\
&=&
\alpha \cdot (\beta \cdot v).
\end{eqnarray*}
Also, if $e$ is the identity element of $G$, then
\begin{eqnarray*}
e \cdot v = \rho(e)v = Iv = v.
\end{eqnarray*}
Hence $V$ is a left $k[G]$-module.

Conversely, suppose $V$ is a left $k[G]$-module. For each $g \in G$, define
\begin{eqnarray}
\rho(g):V \to V, \qquad \rho(g)v = g \cdot v.
\end{eqnarray}
Each $\rho(g)$ is linear, and for $g,h \in G$ and $v \in V$,
\begin{eqnarray*}
\rho(gh)v = (gh)\cdot v = g \cdot (h \cdot v) = \rho(g)\rho(h)v.
\end{eqnarray*}
Thus,
\begin{eqnarray*}
\rho(gh)=\rho(g)\rho(h), \qquad \rho(e)=I,
\end{eqnarray*}
and each $\rho(g)$ is invertible with inverse $\rho(g^{-1})$. Therefore,
\begin{eqnarray}
\rho:G \to GL(V)
\end{eqnarray}
is a linear representation of $G$ on $V$.

Thus, a linear representation of $G$ on $V$ is equivalent to a left
$k[G]$-module structure on $V$.
\end{proof}

Proposition \ref{prop_group_algebra_module} is a standard result in representation
theory \cite{dummit_foote}. It shows that a linear action of a group is not
merely a family of compatible transformations, but carries an intrinsic algebraic
structure. Once the action of $G$ is known, it extends linearly to the group
algebra $k[G]$, thereby endowing the state space with the structure of a module
over $k[G]$.

This becomes especially concrete when the group is generated by a single
element. In that case, the entire action is determined by one operator and its
powers, and the corresponding group algebra reduces to a polynomial-type ring
\cite{lang_algebra_book,dummit_foote}. For example, when $G = \mathbb{Z}$, one
has $k[G] \cong k[x,x^{-1}]$, and when $G = \mathbb{Z}/T\mathbb{Z}$, one has
$k[G] \cong k[x]/(x^T - 1)$
\cite{lang_algebra_book,dummit_foote}.

Thus, a linear system may be viewed not only as a dynamical system or as an
action of time, but also as a module over the algebra generated by its
evolution. This point of view is useful because different classes of linear
systems arise naturally by changing either the time group, the base field, or
both. In particular, standard linear systems over $\mathbb{R}$ or $\mathbb{C}$,
as well as linear systems over finite fields, appear as special cases of this
general framework.

\section{Systems Over Fields of Characteristic Zero}\label{sec_char_zero}

This section shows that the familiar structure theory of linear systems arises
naturally from the representation-theoretic formulation developed above. In
particular, decompositions of the state space into dynamically meaningful
components correspond to decompositions into invariant subrepresentations, and
the form of these decompositions is determined by the algebraic properties of
the underlying field.

\subsection{Algebraically Closed Fields}

Let the state space be $V = k^n$, where $k$ is an algebraically closed field
of characteristic zero, and consider the system
\begin{eqnarray}
x_{t+1} = A x_t, \qquad A \in GL_n(k).
\end{eqnarray}
This defines a map $\rho : \mathbb{Z} \to GL(V)$ by $\rho(t) = A^t$, and hence
a linear action of the time group on $V$.

\begin{proposition}
Let $k$ be an algebraically closed field of characteristic zero. Then the
representation $(V,\rho)$ decomposes as a direct sum of invariant
subrepresentations
\begin{eqnarray}
V = V_1 \oplus \cdots \oplus V_r,
\end{eqnarray}
where each \(V_i\) is the generalized eigenspace associated with a single
eigenvalue \(\lambda_i\) of \(A\).
\end{proposition}

\begin{proof}
Since $k$ is algebraically closed, the characteristic polynomial of $A$
splits completely. Let $\lambda_1,\dots,\lambda_r$ be the distinct eigenvalues
of $A$. For each $\lambda_i$, define
\begin{eqnarray}
V_i = \ker\big((A - \lambda_i I)^n\big).
\end{eqnarray}
Each $V_i$ is invariant under $A$, since $A$ commutes with every polynomial in
$A$. Hence each $V_i$ is invariant under the action $\rho$.

By the Jordan decomposition theorem
\cite{serre_linear_representation_book},
\begin{eqnarray*}
V = V_1 \oplus \cdots \oplus V_r.
\end{eqnarray*}
Thus the representation decomposes into invariant subrepresentations.
\end{proof}

In particular, this shows that the familiar spectral decomposition of \(A\) is, intrinsically, a decomposition of the underlying time representation into invariant subrepresentations. Thus, over an algebraically closed field, the representation decomposes into
invariant subrepresentations corresponding to the generalized eigenspaces of
the evolution operator. In this setting, the classical spectral decomposition
of a linear system is precisely the decomposition of the time action into its
invariant components.

\subsection{Non-Algebraically Closed Fields}

When the base field $k$ is not algebraically closed, the structure changes
fundamentally. The operator $A$ may fail to admit eigenvalues in $k$, and the
system need not decompose into one-dimensional invariant subrepresentations.
Instead, the appropriate decomposition is determined by the factorization of the
minimal polynomial.

Let
\begin{eqnarray}
m_A(t) = p_1(t)^{e_1} \cdots p_r(t)^{e_r}
\end{eqnarray}
be the factorization of the minimal polynomial of $A$ into powers of distinct
irreducible polynomials over $k$ \cite{lang_algebra_book}. For each irreducible
factor $p_i(t)$, define
\begin{eqnarray}
V_i = \ker\big(p_i(A)^{e_i}\big).
\end{eqnarray}

\begin{proposition}
The representation $(V,\rho)$ decomposes as
\begin{eqnarray}
V = V_1 \oplus \cdots \oplus V_r,
\end{eqnarray}
where each $V_i$ is an invariant subrepresentation.
\end{proposition}

\begin{proof}
Since $A$ commutes with every polynomial in $A$, it follows that if $v \in V_i$,
then
\begin{eqnarray}
p_i(A)^{e_i}(Av) = A\,p_i(A)^{e_i}v = 0,
\end{eqnarray}
and hence $Av \in V_i$. Thus each $V_i$ is invariant under $A$, and therefore
under the action $\rho$.

To establish the decomposition, note that the polynomials
$p_1(t)^{e_1},\dots,p_r(t)^{e_r}$ are pairwise coprime. Hence, by the Chinese
remainder theorem, there exist polynomials $q_1(t),\dots,q_r(t)$ such that
\begin{eqnarray}
q_1(t)p_1(t)^{e_1} + \cdots + q_r(t)p_r(t)^{e_r} = 1.
\end{eqnarray}
Evaluating at $A$ yields
\begin{eqnarray}
q_1(A)p_1(A)^{e_1} + \cdots + q_r(A)p_r(A)^{e_r} = I.
\end{eqnarray}
This implies that every $v \in V$ may be written as a sum of vectors in the
subspaces $V_i$, so that
\begin{eqnarray*}
V = V_1 + \cdots + V_r.
\end{eqnarray*}
The sum is direct by the standard primary decomposition argument. Hence
\begin{eqnarray}
V = V_1 \oplus \cdots \oplus V_r.
\end{eqnarray}
\end{proof}

Here each $V_i$ is a \emph{primary} invariant subrepresentation associated with a
single irreducible factor of the minimal polynomial. In general, these
subrepresentations do not split further over the base field. Thus, the failure
of complete spectral decomposition is not a property of the dynamics alone, but
a consequence of the algebraic properties of the field.

\subsection{Real and Complex Fields}

The preceding two subsections specialize to the familiar cases
$k=\mathbb{C}$ and $k=\mathbb{R}$. Over $\mathbb{C}$, which is algebraically
closed, the representation decomposes into invariant subrepresentations
corresponding to individual eigenvalues. Over $\mathbb{R}$, however, irreducible
quadratic factors may appear in the minimal polynomial, leading to
higher-dimensional invariant subrepresentations.

\begin{lemma}
Let $A \in GL_n(\mathbb{R})$. Then the representation $(\mathbb{R}^n,\rho)$
decomposes as
\begin{eqnarray}
\mathbb{R}^n = V_1 \oplus \cdots \oplus V_r,
\end{eqnarray}
where each $V_i$ corresponds either to a real linear factor or to an
irreducible quadratic factor of the minimal polynomial of $A$.
\end{lemma}

\begin{proof}
Over $\mathbb{R}$, the minimal polynomial of $A$ factors into real linear
factors and irreducible quadratic factors \cite{lang_algebra_book}. Applying
the primary decomposition theorem from the previous subsection yields a direct
sum decomposition of $\mathbb{R}^n$ into invariant subspaces corresponding to
these factors.
\end{proof}

In this decomposition, real linear factors correspond to invariant
subrepresentations associated with real eigenvalues, while irreducible quadratic
factors correspond to pairs of complex conjugate eigenvalues and give rise to
two-dimensional invariant real subspaces. These are the basic source of planar
rotational or oscillatory behavior in real linear systems, and the quadratic
case admits the familiar planar realization described below.

\begin{lemma}
Let $A \in GL_n(\mathbb{R})$, and suppose $a \pm ib$ is a pair of complex
conjugate eigenvalues of $A$ with $b \neq 0$. Then there exists a
two-dimensional invariant subrepresentation $V$ such that, with respect to a
suitable basis,
\begin{eqnarray}
A|_V =
\begin{bmatrix}
a & b \\
-b & a
\end{bmatrix}.
\end{eqnarray}
\end{lemma}

\begin{proof}
Let $v \in \mathbb{C}^n$ be an eigenvector corresponding to the eigenvalue
$a+ib$, and write $v = u + iw$ with $u,w \in \mathbb{R}^n$. Then
\begin{eqnarray}
Av = (a+ib)v.
\end{eqnarray}
Equating real and imaginary parts gives
\begin{eqnarray}
Au = au - bw, \qquad Aw = bu + aw.
\end{eqnarray}
Thus $\mathrm{span}\{u,w\}$ is invariant under $A$, and hence under the action
$\rho$. Since $b \neq 0$, the vectors $u$ and $w$ are linearly independent, so
this subspace is two-dimensional.
\end{proof}

Thus, over $\mathbb{R}$, linear systems decompose into two fundamental types of
invariant subrepresentations: those associated with real eigenvalues, and those
associated with complex conjugate pairs. From the representation-theoretic point
of view, the familiar block structure of real linear systems is therefore a
direct consequence of the factorization properties of the base field.

\section{Linear Systems on Finite Fields $\mathbb{F}_p$}\label{sec_finite_fields}

\subsection{Systems Over Fields of Characteristic Zero vs Finite Fields}

Linear systems over $\mathbb{R}$ or $\mathbb{C}$ are typically understood
through their spectral properties. In particular, for a system
$x_{t+1} = A x_t$, the behavior of trajectories is described in terms of the
eigenvalues of $A$, which determine growth, decay, and asymptotic behavior.
When the matrix is diagonalizable, the system decomposes into independent modes,
and more generally, the Jordan form provides a complete description in terms of
generalized eigenvectors.

This framework relies on two key features of the underlying field. First, one
may pass to an algebraic closure in which the characteristic polynomial splits,
ensuring the existence of eigenvalues and invariant subspaces. Second, the field
carries a topology that allows eigenvalues to be interpreted in terms of
magnitude, thereby giving meaning to growth and decay.

Over a finite field \(\mathbb{F}_p\), however, the classical spectral viewpoint
is no longer structurally decisive. Although one may pass to the algebraic
closure $\overline{\mathbb{F}}_p$, in which eigenvalues always exist, this does
not recover the classical picture. The essential issue is that finite fields
carry no notion of magnitude or topology, so eigenvalues do not encode
asymptotic behavior.

This can be seen even in simple examples. A matrix in $GL_2(\mathbb{F}_p)$ with
irreducible quadratic characteristic polynomial admits no eigenvalues in
$\mathbb{F}_p$, and hence no one-dimensional invariant subspaces over the base
field. Passing to an extension field produces eigenvalues, but does not yield a
more meaningful description of the dynamics.

More fundamentally, the state space $\mathbb{F}_p^n$ is finite. As a consequence,
all trajectories are necessarily periodic, and any eigenvalue (in an algebraic
closure) must be a root of unity. Thus, the range of possible behavior is
severely restricted, and eigenvalues cease to play the same organizing role that
they do over $\mathbb{R}$ or $\mathbb{C}$.

These observations suggest that, over finite fields, the natural structure of a
linear system should be understood not primarily through spectral decomposition,
but through the algebraic structure of its time evolution.

\subsection{Finite Cyclic Structure and Emergence of Representation}

We return to the system
\begin{eqnarray}
x_{t+1} = A x_t, \quad A \in GL_n(\mathbb{F}_p),
\end{eqnarray}
and examine its behavior through iteration. For an initial condition $x_0$, one
has $x_t = A^t x_0$, so that the system is governed by the sequence
\begin{eqnarray}
I, A, A^2, A^3, \dots.
\end{eqnarray}

A fundamental structural property now follows from finiteness.

\begin{lemma}
Let $A \in GL_n(\mathbb{F}_p)$. Then there exists a positive integer $T$ such
that
\begin{eqnarray}
A^T = I.
\end{eqnarray}
Consequently,
\begin{eqnarray}
A^{t+T} = A^t, \qquad x_{t+T} = x_t.
\end{eqnarray}
\end{lemma}

\begin{proof}
Since $GL_n(\mathbb{F}_p)$ is finite, the sequence $\{A^t\}$ must contain
repetitions, so $A^{t_1} = A^{t_2}$ for some $t_1 < t_2$. Hence
$A^{t_2-t_1}=I$, and the result follows.
\end{proof}

Thus, time evolution is inherently periodic, and the natural time index is no
longer $\mathbb{Z}$ itself, but time modulo the period $T$. In other words, the
dynamics are naturally organized by the finite cyclic group
$\mathbb{Z}/T\mathbb{Z}$. Accordingly, the original map
\begin{eqnarray}
\tilde{\rho}:\mathbb{Z} \to GL_n(\mathbb{F}_p), \qquad t \mapsto A^t,
\end{eqnarray}
factors through the quotient $\mathbb{Z}/T\mathbb{Z}$.

\begin{lemma}
Let $A \in GL_n(\mathbb{F}_p)$ satisfy $A^T = I$. Then the map
\begin{eqnarray}
\tilde{\rho}:\mathbb{Z} \to GL_n(\mathbb{F}_p), \qquad t \mapsto A^t,
\end{eqnarray}
factors through the quotient $\mathbb{Z}/T\mathbb{Z}$. That is, there exists a
well-defined homomorphism
\begin{eqnarray}
\rho:\mathbb{Z}/T\mathbb{Z} \to GL_n(\mathbb{F}_p)
\end{eqnarray}
such that the following diagram commutes:
\begin{equation}
\begin{tikzcd}[column sep=large, row sep=large]
\mathbb{Z} \arrow[r, "\tilde{\rho}"] \arrow[d,"\pi"] & GL_n(\mathbb{F}_p) \\
\mathbb{Z}/T\mathbb{Z} \arrow[ur, "\rho"']
\end{tikzcd}
\end{equation}
where $\pi$ is the natural projection.
\end{lemma}

\begin{proof}
If $t_1 \equiv t_2 \pmod{T}$, then $t_1=t_2+kT$ for some $k\in\mathbb{Z}$, and
hence
\begin{eqnarray*}
A^{t_1}=A^{t_2+kT}=A^{t_2}(A^T)^k=A^{t_2}.
\end{eqnarray*}
Thus $\tilde{\rho}$ depends only on the class of $t$ modulo $T$, so it defines
a well-defined map $\rho([t])=A^t$. The homomorphism property follows
immediately from
\begin{eqnarray*}
\rho([t_1+t_2]) = A^{t_1+t_2} = A^{t_1}A^{t_2} = \rho([t_1])\rho([t_2]).
\end{eqnarray*}
\end{proof}

Thus, starting from iteration alone, the system naturally induces an action of
the finite cyclic group $\mathbb{Z}/T\mathbb{Z}$ on $\mathbb{F}_p^n$. In this sense, the finite cyclic time group is not imposed externally, but emerges canonically from the algebra of the base field. 

Now, the evolution of an initial condition $x_0$ may be written as
\begin{eqnarray}
x_t = \rho(t)x_0,
\end{eqnarray}
and its trajectory is precisely the orbit of $x_0$ under this action.

\begin{proposition}
For any initial condition $x_0 \in \mathbb{F}_p^n$, the trajectory
\begin{eqnarray}
\{x_t : t \in \mathbb{Z}\}
\end{eqnarray}
coincides with the orbit
\begin{eqnarray}
\{\rho(t)x_0 : t \in \mathbb{Z}/T\mathbb{Z}\}.
\end{eqnarray}
Moreover, the state space $\mathbb{F}_p^n$ decomposes as a disjoint union of
such trajectories.
\end{proposition}

\begin{proof}
The first statement follows directly from the definition
$x_t=A^t x_0=\rho(t)x_0$. The second follows because the orbits of a group
action form a partition of the underlying set.
\end{proof}

Note that the decomposition into orbits is a geometric consequence of the
induced group action, and should be distinguished from the linear decomposition
of the state space that will arise later from the module-theoretic viewpoint.

\begin{remark}
Although one has an isomorphism
\begin{eqnarray}
\mathbb{Z}/T\mathbb{Z} \cong \mathbb{Z}/p_1^{k_1}\mathbb{Z} \times \cdots \times \mathbb{Z}/p_r^{k_r}\mathbb{Z},
\end{eqnarray}
this does not imply that the dynamics decompose into independent components.
The system is generated by a single element $A$, and hence evolves along a
single cyclic time direction. The above decomposition reflects the arithmetic
structure of the time group, but does not induce a decomposition of the
dynamics.
\end{remark}

\section{Algebraic Structure of Linear Systems over Finite Fields}\label{sec_algebraic_structure}

The representation-theoretic viewpoint implies that the dynamics of the system
are generated by a single linear operator, namely the state transition matrix
$A$. This naturally leads to an algebraic extension of the dynamics by
considering not only powers of $A$, but also polynomial expressions in $A$.

\subsection{Polynomial Action and Induced Algebraic Structure}

For any polynomial
\begin{eqnarray}
p(x) = c_0 + c_1 x + \cdots + c_k x^k \in \mathbb{F}_p[x],
\end{eqnarray}
we define the associated operator
\begin{eqnarray}
p(A) = c_0 I + c_1 A + \cdots + c_k A^k.
\end{eqnarray}
Thus, each polynomial gives rise to a linear operator via the substitution
$x \mapsto A$.

This construction is compatible with the algebraic structure of
$\mathbb{F}_p[x]$, as formalized in the following result.

\begin{proposition}
The map
\begin{eqnarray}
\Phi : \mathbb{F}_p[x] \to \mathrm{End}(\mathbb{F}_p^n), \quad
p(x) \mapsto p(A),
\end{eqnarray}
is a homomorphism of $\mathbb{F}_p$-algebras.
\end{proposition}

\begin{proof}
For any $p(x), q(x) \in \mathbb{F}_p[x]$ and $\alpha \in \mathbb{F}_p$, we have
\begin{eqnarray*}
\Phi(p + q) &=& p(A) + q(A), \qquad \Phi(\alpha p) = \alpha p(A), \\
\Phi(pq) &=& p(A)q(A),
\end{eqnarray*}
where the last identity follows from $A^i A^j = A^{i+j}$. Hence $\Phi$
preserves addition, scalar multiplication, and multiplication, and is therefore
an algebra homomorphism.
\end{proof}

Thus, the polynomial ring $\mathbb{F}_p[x]$ acts on the state space via the
evaluation map $x \mapsto A$, endowing $\mathbb{F}_p^n$ with the structure of
an $\mathbb{F}_p[x]$-module.

From this perspective, the dynamics are no longer described solely in terms of
iteration of $A$, but through the action of the polynomial algebra generated by
the time-shift operator.

\subsection{Equivalence with Module Structure over a Polynomial Ring}

In the previous section, we showed that a linear system over
$\mathbb{F}_p$ with $A^T = I$ naturally defines a representation
\begin{eqnarray}
\rho : \mathbb{Z}/T\mathbb{Z} \to GL(\mathbb{F}_p^n),
\end{eqnarray}
given by $\rho(t) = A^t$.

By Proposition \ref{prop_group_algebra_module}, this representation is
equivalent to endowing the state space $\mathbb{F}_p^n$ with the structure
of a module over the group algebra $\mathbb{F}_p[\mathbb{Z}/T\mathbb{Z}]$.

The following theorem makes this structure explicit.

\begin{theorem}
Let $A \in GL_n(\mathbb{F}_p)$ satisfy $A^T = I$, and consider the linear
system $x_{t+1} = Ax_t$. Then the state space $\mathbb{F}_p^n$ is naturally
a module over the ring
\begin{eqnarray}
\mathbb{F}_p[x]/(x^T - 1),
\end{eqnarray}
where the action is given by
\begin{eqnarray}
x \cdot v = Av,
\end{eqnarray}
and extended linearly to all of $\mathbb{F}_p[x]/(x^T - 1)$.

Conversely, any module over $\mathbb{F}_p[x]/(x^T - 1)$ defines a linear
system over $\mathbb{F}_p$ whose evolution operator satisfies $A^T = I$.
\end{theorem}

\begin{proof}
The group algebra of the finite cyclic group
$\mathbb{Z}/T\mathbb{Z}$ over $\mathbb{F}_p$ is isomorphic to the quotient ring
\begin{eqnarray}
\mathbb{F}_p[\mathbb{Z}/T\mathbb{Z}] \cong \mathbb{F}_p[x]/(x^T - 1),
\end{eqnarray}
where the generator of $\mathbb{Z}/T\mathbb{Z}$ corresponds to the indeterminate
$x$ \cite{dummit_foote}.

Thus, by Proposition \ref{prop_group_algebra_module}, giving a representation
$\rho : \mathbb{Z}/T\mathbb{Z} \to GL(\mathbb{F}_p^n)$ is equivalent to giving
$\mathbb{F}_p^n$ the structure of a module over
$\mathbb{F}_p[x]/(x^T - 1)$.

Under this identification, the action of $x$ corresponds to the operator
$A = \rho(1)$, and hence $x \cdot v = Av$.

Conversely, suppose $V = \mathbb{F}_p^n$ is a module over
$\mathbb{F}_p[x]/(x^T - 1)$. Define a linear operator $A : V \to V$ by
\begin{eqnarray}
A v := x \cdot v.
\end{eqnarray}
Since $x^T = 1$ in $\mathbb{F}_p[x]/(x^T - 1)$, it follows that
\begin{eqnarray}
A^T v = x^T \cdot v = v
\end{eqnarray}
for all $v \in V$, and hence $A^T = I$.

Now define $\rho : \mathbb{Z}/T\mathbb{Z} \to GL(V)$ by
\begin{eqnarray}
\rho(t) = A^t.
\end{eqnarray}
Then $\rho$ is a group homomorphism, and the corresponding linear system
$x_{t+1} = Ax_t$ recovers the original module action.
\end{proof}

Thus, a linear system over $\mathbb{F}_p$ is equivalently a module over the
ring $\mathbb{F}_p[x]/(x^T - 1)$, where the indeterminate $x$ acts as the
time-shift operator. This shows that the algebraic structure of the system is
completely determined by the representation of the finite cyclic time group. In this sense, the module structure is not an auxiliary construction, but the natural algebraic form of the underlying time representation.

\subsection{Module-Theoretic Interpretation of System Structure}

The module structure provides a systematic framework for analyzing the system.
In particular, invariant subspaces of the system correspond precisely to
submodules of $\mathbb{F}_p^n$, and decomposition of the state space into
independent subsystems corresponds to decomposition into submodules.

Moreover, since the action is generated by a single linear operator \(A\), the
state space may be viewed as a finitely generated \(\mathbb{F}_p[x]\)-module,
where \(x\) acts by \(A\). Since \(\mathbb{F}_p[x]\) is a principal ideal domain,
the structure theorem for finitely generated modules applies.

\begin{proposition}
Let $V = \mathbb{F}_p^n$ be the state space of the system
$x_{t+1} = Ax_t$. Then, viewed as an $\mathbb{F}_p[x]$-module via the action
$x \cdot v = Av$, one has a decomposition
\begin{eqnarray}
V \cong \bigoplus_{i=1}^r \mathbb{F}_p[x]/(f_i(x)),
\end{eqnarray}
where each $f_i(x)$ is a monic polynomial and
\begin{eqnarray}
f_1(x)\mid f_2(x)\mid \cdots \mid f_r(x).
\end{eqnarray}
If, in addition, $A^T = I$, then each $f_i(x)$ divides $x^T - 1$.
\end{proposition}

\begin{proof}
This follows from the structure theorem for finitely generated modules over
principal ideal domains \cite{dummit_foote,lang_algebra_book}. The divisibility
condition \(f_i(x)\mid x^T-1\) follows from the fact that \(x^T-1\) annihilates
the module.
\end{proof}

Each summand $\mathbb{F}_p[x]/(f_i(x))$ corresponds to a cyclic submodule
generated by a single state vector, whose evolution is governed by the relation
$f_i(A)=0$. From a dynamical perspective, this decomposition plays a role analogous to the
Jordan decomposition in classical linear systems theory. However, instead of
being organized around eigenvalues and generalized eigenvectors, the structure
is determined by invariant factors and annihilating polynomials.

Thus, over finite fields, the global dynamics of the system may be understood as
a direct sum of independent cyclic subsystems, each characterized by a
polynomial relation satisfied by the evolution operator.

\section{Linear Systems as Representations of Time}\label{sec_linear_system_as_representation}

The preceding developments show that the module-theoretic structure of a linear
system is not an additional assumption, but a direct consequence of viewing time
evolution as an action. In classical linear systems theory, one typically begins
with a linear operator $A$ and studies its properties, often introducing
polynomial actions and invariant factors as analytical tools. In contrast, the
viewpoint adopted here reverses this logic: the evolution of the system defines
an action of time, and the associated algebraic structures arise naturally from
this action.

In the discrete-time setting, where the time group is $\mathbb{Z}$, the
corresponding group algebra over a field $k$ is isomorphic to the Laurent
polynomial ring $k[x,x^{-1}]$ \cite{dummit_foote,lang_algebra_book}, where the
indeterminate $x$ represents the unit time shift. Under this identification,
the state space becomes a module over $k[x,x^{-1}]$, with the action
$x \cdot v = Av$.

Similarly, when the time group is $\mathbb{Z}/T\mathbb{Z}$, one has
\begin{eqnarray}
k[\mathbb{Z}/T\mathbb{Z}] \cong k[x]/(x^T - 1),
\end{eqnarray}
so the state space becomes a module over $k[x]/(x^T - 1)$, reflecting the
periodic nature of the dynamics.

Thus, the algebraic structure of a linear system is intrinsic: the state space
is naturally a module over an algebra generated by time, and classical
constructions such as annihilating polynomials and invariant factors arise as
structural features of this module. This leads to the following intrinsic,
coordinate-free definition of a linear dynamical system.

\begin{definition}[\textbf{Linear System}]
Let $G$ be a group and let $V$ be a vector space over a field $k$. A
\emph{linear dynamical system with time group $G$} is a group homomorphism
\begin{eqnarray}
\rho : G \to GL_k(V),
\end{eqnarray}
where $GL_k(V)$ denotes the group of invertible linear operators on $V$.
\end{definition}

In this formulation, the evolution of a state $v \in V$ is given by
\begin{eqnarray}
v_t = \rho(t)v, \quad t \in G,
\end{eqnarray}
and trajectories correspond to orbits of the group action.

\section{Discussions}\label{sec_discussion}

The central point of this paper is that a linear system is most naturally
understood not merely as a matrix, but as an action of time on a vector space.
From this perspective, the matrix is best viewed as the concrete generator of a
representation, rather than as the full object in isolation. Once this shift in
viewpoint is made, many familiar aspects of linear systems theory reorganize
themselves in a surprisingly natural way.

One consequence is that the structure of a linear system is governed by two
ingredients: the algebra of the time group and the algebraic properties of the
underlying field. The time group determines how evolution is organized, while
the field determines which decompositions, factorizations, and invariant
structures are available. From this perspective, many of the usual distinctions
between real, complex, and finite-field systems appear not as separate
phenomena, but as different manifestations of the same underlying principle.

This viewpoint also clarifies the role of spectral methods. Over fields such as
$\mathbb{R}$ and $\mathbb{C}$, spectral theory provides a convenient and often
powerful language for describing the resulting structure. But it is not the
foundation of the theory. Rather, it is one especially successful expression of
a more basic algebraic framework. In settings such as finite fields, where
eigenvalues do not play the same structural role, the underlying module
structure remains fully intact and, in some ways, becomes easier to see.

Viewed this way, the natural object of study in linear systems is not the
matrix in isolation, but the representation of time that it induces. Once this
idea is taken seriously, several natural extensions come into view. Allowing the
time group to be non-abelian leads toward switched or structured systems
\cite{swtiched_system_non_abelian_time}; passing from vector spaces to spaces
built from them leads toward induced nonlinear dynamics
\cite{dani_homogeneous_space,kleeman_linear_algebra_dynamical_systems}; and
replacing discrete time by continuous time leads naturally to the setting of Lie
groups and their infinitesimal generators
\cite{hall_lie_group_representation_book}. 

None of this diminishes the importance of matrices. They remain the concrete
carriers of the dynamics. But they are perhaps best understood as coordinate
representatives of a more structural object. In this sense, the present viewpoint does not replace the classical matrix description, but explains why it works as well as it does. Once that structure is made explicit, linear systems begin to look less like isolated operators and more
like what they fundamentally are: representations of time.

\section{Conclusions}\label{sec_conclusions}

In this paper, we developed a viewpoint that provides a unified description of
different classes of linear systems. In particular, we formulated a
representation-theoretic definition of linear systems by viewing them as actions
of time groups on vector spaces. From this perspective, the structure of a
linear system is determined not only by its evolution operator, but more
fundamentally by the algebra of time and the field over which the system is
defined.

Within this setting, familiar features of linear systems—such as invariant
subspaces and spectral decomposition—arise as manifestations of an underlying
algebraic structure. In settings where spectral methods are not natural, such
as finite fields, this algebraic structure remains fully intact and provides a
natural alternative description. More broadly, the framework suggests that many
systems which appear different in the classical literature are, at a deeper
level, instances of the same underlying idea: a linear system is a
representation of time.

\bibliographystyle{IEEEtran}
\bibliography{references}

\end{document}